\documentclass[aps,prl,twocolumn,superscriptaddress]{revtex4}
\usepackage{bm}
\usepackage{epsfig}
\usepackage[english]{babel}
\usepackage{graphicx,times}
\usepackage{amssymb,amsmath,bbm,amsfonts,amsthm,subfigure,dsfont}

\newcommand{\abs}[1]{\left\vert#1\right\vert}

\newcommand{\bra}[1]{\left\langle#1\right\vert}
\newcommand{\ket}[1]{\left\vert#1\right\rangle}
\newcommand{\braket}[2]{\left.\left\langle#1\right|#2\right\rangle}
\newcommand{\Tr}[1]{\mbox{Tr}[#1]}

\newcommand{\DStoA}{\overrightarrow{\cal D}_{SA}}
\newcommand{\DAtoS}{\overleftarrow{\cal D}_{SA}}

\newtheorem{theorem}{\bf{Theorem}}

\newtheorem{lemma}{\bf{Lemma}}
\newtheorem{corollary}{\bf{Corollary}}

\begin{document}

\title {Daemonic Ergotropy: Enhanced Work Extraction from Quantum Correlations}
\author{G.~Francica}
\affiliation{Dip.  Fisica, Universit\`a della Calabria, 87036
Arcavacata di Rende (CS), Italy} \affiliation{INFN - Gruppo
collegato di Cosenza, Cosenza Italy}

\author{J.~Goold}
\affiliation{International Centre for Theoretical Physics, Trieste Italy}
\author{M.~Paternostro}
\affiliation{Centre for Theoretical Atomic, Molecular and Optical Physics, School of Mathematics and Physics, Queen's University,
Belfast BT7 1NN, United Kingdom}
\author{F.~Plastina}
\affiliation{Dip.  Fisica, Universit\`a della Calabria, 87036
Arcavacata di Rende (CS), Italy} \affiliation{INFN - Gruppo
collegato di Cosenza, Cosenza Italy}

\begin{abstract}
We investigate how the presence of quantum correlations can
influence work extraction in closed quantum systems, establishing
a new link between the field of quantum non-equilibrium
thermodynamics and the one of quantum information theory. We
consider a bipartite quantum system and we show that it is
possible to optimise the process of work extraction, thanks to the
correlations between the two parts of the system, by using an
appropriate feedback protocol based on the concept of ergotropy.
We prove that the maximum gain in the extracted work  is related
to the existence of quantum correlations between the two parts, quantified by either quantum discord or, for pure states, entanglement. We then  illustrate our general findings on a simple
physical situation consisting of a qubit system.
\end{abstract}
\maketitle

The thermodynamic implications of
quantum dynamics are currently helping us build new architectures
for the super-efficient nano- and micro-engines, and design
protocols for the manipulation and management of work and heat
above and beyond the possibilities offered by merely classical
processes~\cite{review}. Exciting experimental progress towards
the achievement of such paramount goals is currently
ongoing~\cite{rossnagel}. Quantum coherences are believed to be
responsible for the extraction of work from a single heat
bath~\cite{scully}, while weakly driven quantum heat engines are
known to exhibit enhanced power outputs with respect to their
classical (stochastic) versions~\cite{uzdin}.

Despite such evidences, the identification of the specific
features of quantum systems that might influence their
thermodynamic performance is currently a debated point. In
particular, the role that quantum correlations and coherences in
schemes for the extraction of work from quantum systems appears to
be quite controversial~\cite{acin,fusco}. Yet, the clarification
of the relevance of genuinely quantum features would be key for
the grounding of quantum thermodynamics as a viable route towards
the construction of a framework for quantum
technologies~\cite{review}. Indeed, the very tight link between
thermodynamics and quantum entanglement~\cite{pleniovitelli} cries
loud for the clarification of the role of quantum correlations as
a resource for coherent thermodynamic processes and
transformations~\cite{oppenheim}.

 \begin{figure}[b]
        \begin{center}
        \includegraphics[width=0.8\columnwidth]{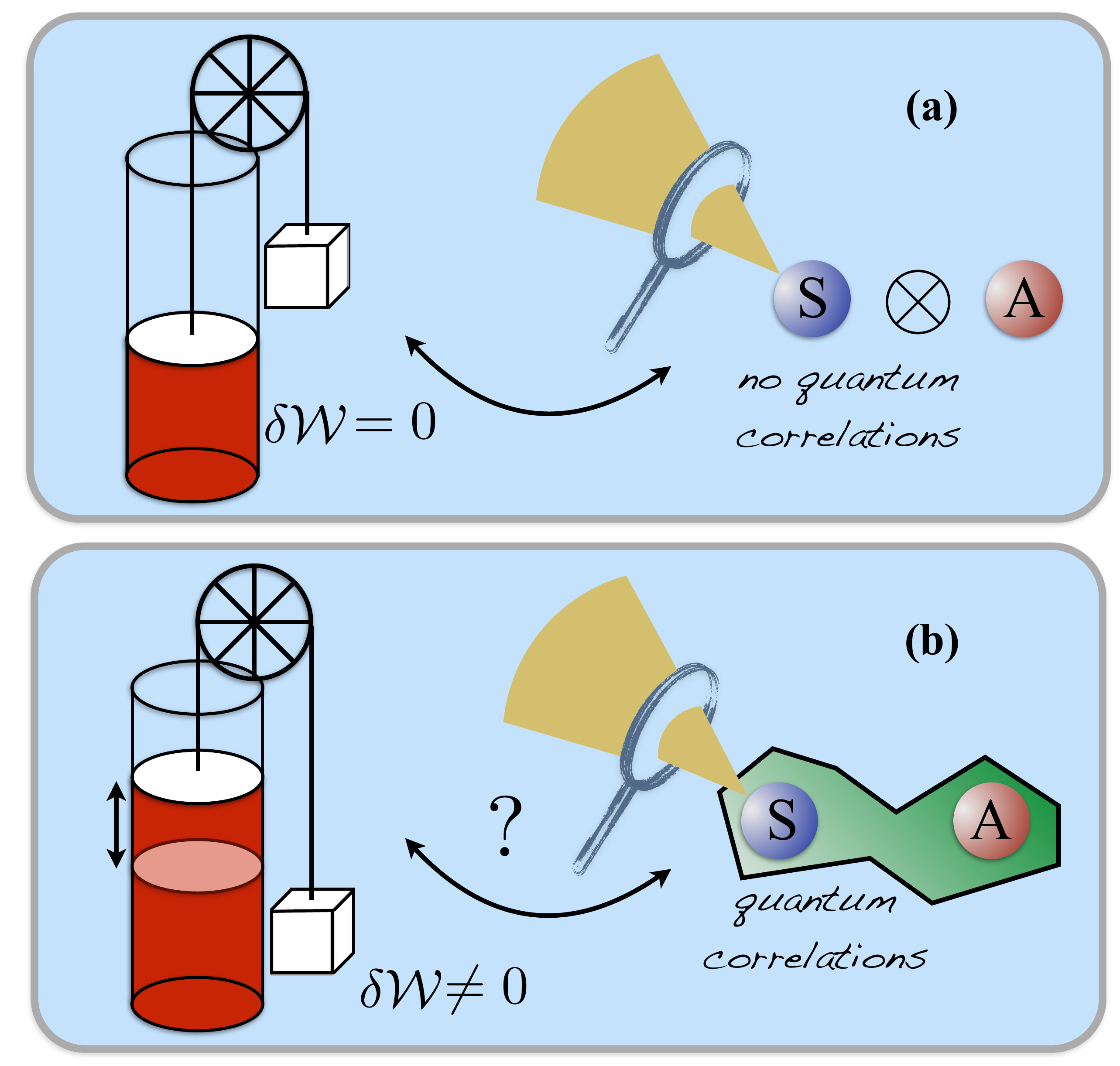}
\caption{(Color online) {\bf (a)} In our ancilla-assisted protocol, a null daemonic gain gain [i.e. $\delta{\cal W}=0$, cf. Eq.~\eqref{daemonergo}] implies the absence of quantum correlations between the system $S$ and the ancilla $A$ [as measured by the discord associated with measurements on $S$, cf. Eq.~\eqref{DStoA}]. {\bf (b)} A non-null value of $\delta{\cal W}$, on the other hand, implies the possible existence of quantum correlations between $S$ and $A$. For pure bipartite states (in arbitrary dimensions), the nullity of the daemonic gain is a necessary and sufficient condition for separability.}
        \label{scheme}
        \end{center}
\end{figure}

In this paper, we make steps towards the clarification of the role
of quantum correlations in work extraction processes by
investigating a simple ancilla-assisted primitive. We address the
concept of {\it ergotropy}, i.e. the maximum work that can be
gained from a quantum state, with respect to some reference
Hamiltonian, under cyclic unitaries~\cite{allahverdyan}. We
consider the joint state of a system and an isodimensional
ancilla, which can be measured in an arbitrary basis, and show
that quantum correlations are related to a possible increase of
the extracted work. More precisely, we demonstrate that if system
and ancilla share no quantum quantum discord~\cite{discord}, then
the information gathered through the measurements performed on the
state of the ancilla cannot help in catalyzing the extraction of
work from the system. We extend this result to the case of quantum
entanglement, thus establishing a tight link between enhanced
work-extraction performances and a clear-cut resource in quantum
information processing. We illustrate our findings for the
relevant case where system and ancilla are both embodied by
qubits, showing the existence of a families of states that provide
attainable (upper and lower) bounds to the {\it gain} in
extractable work at a set degree of quantum correlations between
system and ancilla. Not only do our results shed light on the core
role that quantum correlations have in thermodynamically relevant
processes they also open up the pathway towards the study of the
implications of the structure of generally quantum correlated
resources for ancilla-assisted work extraction schemes and the
grounding of the technological potential of the thermodynamics of
quantum systems.

 \noindent
{\it Ergotropy.--} We start by introducing the {\it ergotropy},
which is the maximum amount of work  that can be extracted from a
quantum system in a given state by means of a cyclical unitary
transformation~\cite{allahverdyan}. Consider a system $S$ with
Hamiltonian $\hat H_S$ and density matrix $\hat\rho_S$ given by,
\begin{equation}
  \hat H_S = \sum_k \epsilon_k \ket{\epsilon_k}\bra{\epsilon_k},\quad \hat \rho_S = \sum_k r_k \ket{r_k} \bra{r_k},
\end{equation}
with $\epsilon_k$ the energy of the $k^{\rm th}$ eigenstate
$\ket{\epsilon_k}$ of $\hat H_S$ and $r_k$ the population of the
eigenstate $\ket{r_k}$ of $\hat\rho_S$. If $\rho_S$ is a {\it
passive} state (i.e., if $[\hat\rho_s,\hat H_s]=0$ and  $r_n \geq
r_{m}$ whenever $\epsilon_n < \epsilon_m$), no work can be
extracted by means of a cyclical variation of the Hamiltonian
parameters ($\hat H_S(0)=\hat H_S(\tau)=\hat H_S$) over a fixed
time interval $[0,\tau]$ \cite{Pusz:78,Lenard:78,Allahverdyan:02}.
If the initial state $\hat\rho_S$ is not passive with respect to
$\hat H_S$, then work may be extracted cyclically, and its maximal
amount, the ergotropy $\mathcal{W}$, has been shown by
Allahverdyan to be given by
\begin{equation}
  \mathcal{W} = \sum_{j,k} r_k \epsilon_j \left( \abs{\braket{\epsilon_j}{r_k}}^2-\delta_{jk}\right).
\end{equation}

\noindent
{\it Daemonic work and quantum correlations.--} 
In order to connect with the theory of quantum correlations, we extend the framework for maximal work extraction by introducing a non-interacting ancilla $A$ and assume that system and ancilla are initially prepared in the joint state $\hat\rho_{SA}$. The intuition behind the protocol, that will be discussed below, is that should $\hat\rho_{SA}$ bring about correlations between $S$ and $A$, a measurement performed on the ancilla would give us information about the state of $S$, which could then be used to enhance the amount of work that can be extracted from its state.

Within such a generalized framework, the amount of extractable
work crucially depends on the measurements performed on $A$, that
we describe through a set of orthogonal projectors
$\{\hat\Pi^A_a\}$. Upon the measurement of $A$ with outcome $a$,
the state of the system collapses onto the conditional density
matrix $\hat\rho_{S|a} =
\text{Tr}_A[\hat\Pi^A_a\hat\rho_{SA}\hat\Pi^A_a]/{p_a}$ with
probability $p_a = \Tr{\hat\Pi^A_a\hat\rho_{SA}}$. The time
evolution of state $\rho_{S|a}$ then follows a cyclic unitary
process $\hat U_a$ conditioned on the outcome of the measurement.
By averaging over all of the possible outcomes of the measurement,
the work extracted from the state of $S$ reads
 \begin{equation}
  W_{\{\hat\Pi^A_a\}} = \text{Tr}[ \hat \rho_S \hat H_S ] - \sum_a p_a \text{Tr}[\hat U_a \hat\rho_{S|a} \hat U_a^\dag \hat H_S ]
\end{equation}
with $\hat \rho_S=\text{Tr}_A[\hat\rho_{SA}]$. This quantity
explicitly depends on the specific control strategy determined by
the outcomes of the measurements $\{\hat\Pi^A_a\}$. We can thus
proceed to maximise the extracted work by performing the optimal
ergotropic transformation for each of the $\hat\rho_{S|a}$ such
that
\begin{equation}
  \mathcal{W}_{\{\hat\Pi^A_a\}} = \text{Tr} [\hat \rho_S \hat H_S ] - \sum_a p_a \sum_k r_k^a \epsilon_k
\end{equation}
with $\{\hat\Pi^A_a\}$ a set of orthogonal projective
measurements, and $r_k^a$ the eigenvalues of $\hat\rho_{S|a}$. We
call this quantity the {\it Daemonic Ergotropy}.

On the other hand if we do not use the information obtained upon
measuring the ancilla, and thus control the system in the same
way, independently of the measurement outcomes (i.e. $\hat
U_a=\hat U$ for any $a$), the maximum extractable work would be
given by the ergotropy $\mathcal{W}$ associated with state
$\hat\rho_S=  \text{Tr}_A\{\rho_{SA} \} = \sum_k r_k \ket{r_k}\bra{r_k}$.
In the Appendix we have shown that the information acquired
through the measurements allows to extract more work than in the
absence of them, that is $\mathcal{W}_{\{\hat\Pi^A_a\}}\geq
\mathcal{W}$ and that it provides a tighter upperbound on the
ergotropy than the one derived in \cite{allahverdyan}. If we call
$\mathcal{W}_\text{th}$ the work extracted when the final state is
the Gibbs state $e^{-\beta\hat H_S(\lambda_0)}/\Tr{e^{-\beta\hat
H_S(\lambda_0)}}$ with the same entropy as $\hat\rho_S$, then
$\mathcal{W}_\text{th} \geq \mathcal{W}_{\{\hat\Pi^A_a\}}\geq
\mathcal{W}$.
The characterization of the efficiency of work extraction scheme,
though, should take into account the energetic cost of the
measurements $\Delta E_\text{meas}$, whose quantification depends
on several factors. However, it cannot be smaller than the average
variation in the energy of $A$, so that a lower value can be
established as $\Delta E_\text{meas} \geq \sum_a p_a
\Tr{\hat\Pi^A_a\hat H_A} - \text{Tr}[\hat H_A \hat\rho_A]$ with
$\hat H_A$ the Hamiltonian of the ancilla and $\hat
\rho_A=\text{Tr}_S[\hat\rho_{SA}]$ its reduced state.

In what follows the main object of our attention will be the
difference $\mathcal{W}_{\{\Pi^A_a\}} - \mathcal{W}$, which is
expected to be related to the (nature and degree of) correlations
between $S$ and $A$. For instance, should $S$ and $A$ be initially
statistically independent, i.e.
$\hat\rho_{SA}=\hat\rho_S\otimes\hat\rho_A$, the measurements on
the ancilla would not bring about any information on the state of
$S$, as we would have $\hat\rho_{S|a}=\hat\rho_S$ for any set
$\{\hat\Pi^A_a\}$ and outcome $a$. Consequently, there would be no
gain in work extraction and $\mathcal{W}_{\{\Pi^A_a\}}=
\mathcal{W}$. However, besides such a rather extreme case, other
instances of no gain in work extraction (from correlated
$\hat\rho_{SA}$ states) might be possible, and our goal here is to
characterize such occurrences.

In order to achieve this goal, we introduce the quantity
\begin{equation}
\label{daemonergo}
\delta \mathcal{W} = \text{max}_{\{\hat\Pi_a^A\}} \mathcal{W}_{\{\hat\Pi^A_a\}} - \mathcal{W},
\end{equation}
which we dub, from now on {\it daemonic gain} in light of its
ancilla-assisted nature. Clearly, $\delta{\cal W}\ge0$ because of
the considerations above and the optimization entailed in
Eq.~\eqref{daemonergo}.

Our aim is to connect $\delta \mathcal{W}$ to quantum
correlations. To this end, we notice that $\delta \mathcal{W}$ is
invariant under local unitary transformations: any unitary
transformation on $S$ can be incorporated in the transformations
used for the extraction of work, while any unitary on $A$ is
equivalent to a change of measurement basis. Then, we consider
quantum discord~\cite{discord} as the figure of merit to quantify
the degree of quantum correlations shared by system and ancilla.
For measurements performed on the system $S$, discord is defined
as
\begin{equation}
\label{DStoA}
\DStoA={\cal I}_{SA}-\max_{\{\hat\Phi^A_a\}}\overrightarrow{\cal J}_{SA},
\end{equation}
where ${\cal I}_{SA}$ is the mutual information between $S$ and
$A$, and $\overrightarrow{\cal J}_{SA}$ is the one-way classical
information associated with an orthogonal measurement set
$\{\hat\Phi^A_a\}$ performed on the system~\cite{discord}.
Explicit definitions are given in the Appendix. We believe the
choice of Eq.~\eqref{DStoA} is well motivated in light of the
explicit asymmetry of both $\delta{\cal W}$ and $\DStoA$ with
respect to the subject of the projective measurements.
We are now in a position to state one of the main results of our
work, which we present in the form of the following Theorem:
 \begin{theorem}
 \label{th1}
For any system $S$ and ancilla $A$ prepared in a state $\hat\rho_{SA}$, we have
\begin{equation}
  \delta \mathcal{W} = 0 \Rightarrow  \overrightarrow{\mathcal{D}}_{SA}=0
\end{equation}
with $\delta{\cal W}$ and $\DStoA$ as defined in Eq.~\eqref{daemonergo} and \eqref{DStoA}, respectively.
\end{theorem}
The asymmetry of the daemonic gain is well reflected into the impossibility of linking $\delta{\cal W}$ to the discord associated with measurements performed on the ancilla. That is
 \begin{equation}
 \label{corol}
   \delta \mathcal{W} = 0 \nRightarrow  \DAtoS=0.
 \end{equation}
The proof of both Theorem~\ref{th1} and the corollary statement in Eq.~\eqref{corol} are presented fully in the Appendix, while a scheme of principle is presented in Fig.~\ref{scheme}. It is important to observe that, in general, the inverse of Theorem~\ref{th1} does not hold, i.e.
$ \overleftarrow{\mathcal{D}}_{SA}=0 \text{ or } \overrightarrow{\mathcal{D}}_{SA}=0 \nRightarrow   \delta \mathcal{W} = 0$ as there can well be classically correlated states associated with a non-null daemonic gain.
However, a remarkable result is found when $\hat\rho_{SA}$ is pure, for which the only possible quantum correlations are embodied by entanglement.
\begin{theorem}
\label{th2}
For any system $S$ and ancilla $A$ prepared in a pure state $\hat \rho_{SA}=\ket{\psi}\bra{\psi}_{SA}$ we have
\begin{equation}
\delta{\cal W}=0\Leftrightarrow\ket{\psi}_{SA}~\text{is separable},
\end{equation}
and $\delta \mathcal{W} = \sum_k r_k \epsilon_k - \epsilon_1$, where $r_k$ are the Schmidt coefficients of $\ket{\psi}_{SA}$ and $\epsilon_k$ are the eigenvalues of $\hat H_S$, ordered such that $r_k\geq r_{k+1}$ and $\epsilon_k\leq \epsilon_{k+1}$.
\end{theorem}
Theorem~\ref{th2} is a thermodynamically motivated separability criterion for pure bipartite states in arbitrary dimensions and an explicit quantitative link between the theory of entanglement and the thermodynamics of information.

\noindent
{\it Illustrations in two-qubit systems.--} The statements in Theorems~\ref{th1} and \ref{th2} are completely general, and independent of the nature of either $S$ or $A$, which could in principle live in Hilbert spaces of different dimensions. However, in order to illustrate their implications and gather further insight into the relation between the introduced daemonic gain and both discord and entanglement, here we focus on the smallest non-trivial situation, which is embodied by a two-qubit system.

\begin{figure}[b]
        \begin{center}
                \includegraphics[width=0.9\columnwidth]{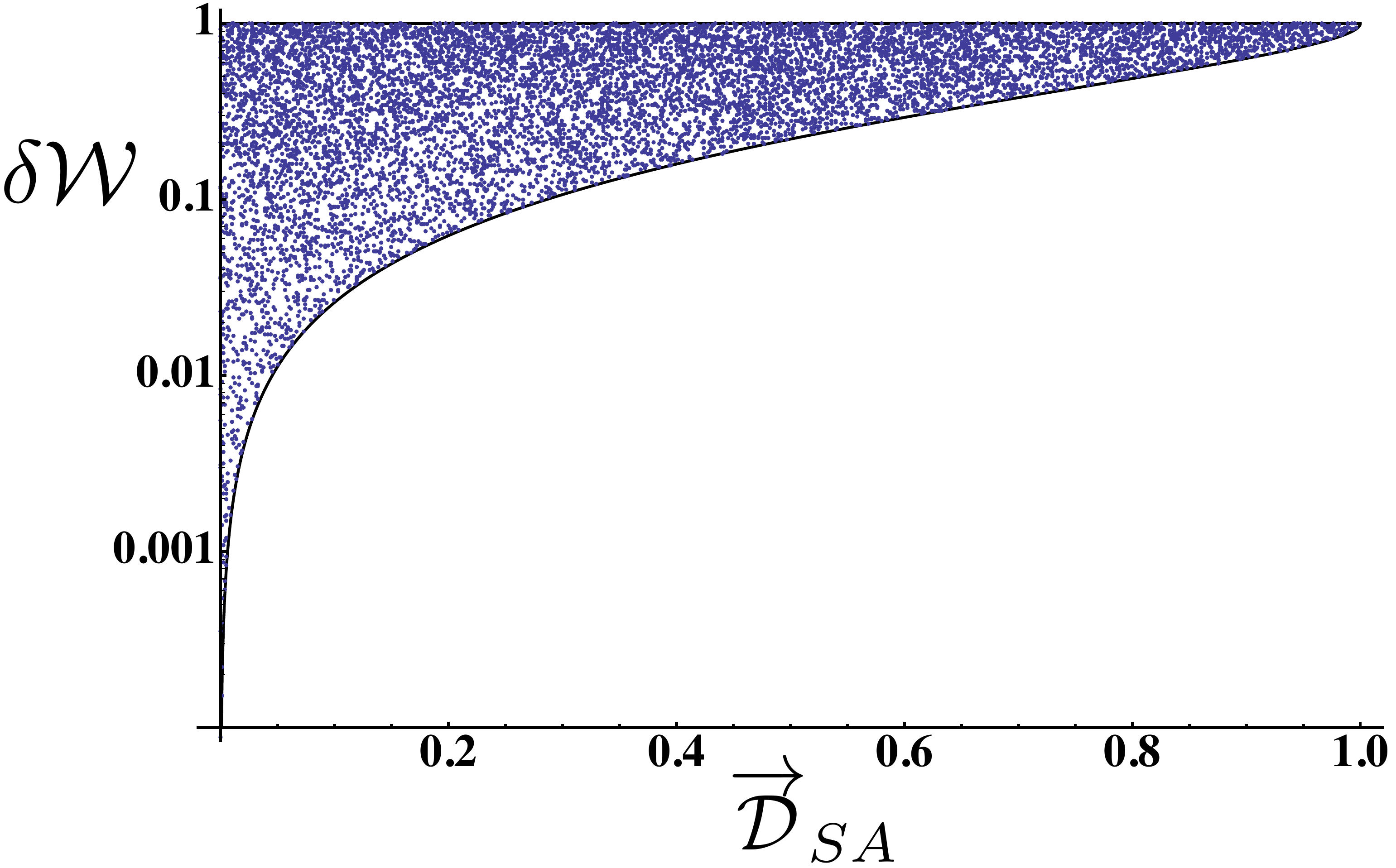}
\caption{(Color online) Distribution of two-qubit states in the
daemonic gain-vs-discord plane. We have generated $3\times10^3$
general random states of system and ancilla, evaluating the
discord and daemonic gain  for each of them (blue dots). The blue
curves enclosing the distribution correspond to the boundaries
discussed in the body of the paper. Notice that states with no
quantum correlations may correspond to arbitrarily large values of
the daemonic gain $\delta{\cal W}$. }
        \label{distDisc}
        \end{center}
\end{figure}

We start with the implications of Theorem~\ref{th1} and compare
$\delta \mathcal{W}$ with discord
$\overrightarrow{\mathcal{D}}_{SA}$. Since both these quantities
are invariant under local unitary transformations on
$\hat\rho_{SA}$, without loss of generality we can consider the
system Hamiltonian $\hat H_S=-\sigma_z$. In Fig.~\ref{distDisc} we
show the distribution of randomly generated two-qubit states over
the $\delta{\cal W}$-versus-$\DStoA$ plane. Such an extensive
numerical analysis reveals that, for any state $\hat\rho_{SA}$
with discord $\overrightarrow{\mathcal{D}}_{SA}=\mathcal{D}$, we
have
\begin{equation}
\delta \mathcal{W}\geq\delta \mathcal{W}_\text{min}(\mathcal{D})=h\left(1-{\mathcal{D}}/{2}\right),
\end{equation}
where $h(x)=-x\log_2(x)-(1-x)\log_2(1-x)$. The monotonicity of
$h(x)$ implies that growing values of quantum correlations are
associated with a monotonically increasing daemonic gain: for the
states lying on such lower bound, quantum correlations form a
genuine resource for the catalysis of thermodynamic work
extraction. Moreover, as $\lim_{x\to1}h(x)=0$, a two-qubit system
with $\DStoA=0$ (i.e. a classically correlated state) can achieve,
in principle, any value of daemonic gain up to the maximum that,
for this case, is $\delta{\cal W}=1$. On the other hand, the
daemonic ergotropy is maximized by taking pure two-qubit states
with growing degree of entanglement.


\begin{figure}[t]
        \begin{center}
        \includegraphics[width=0.92\columnwidth]{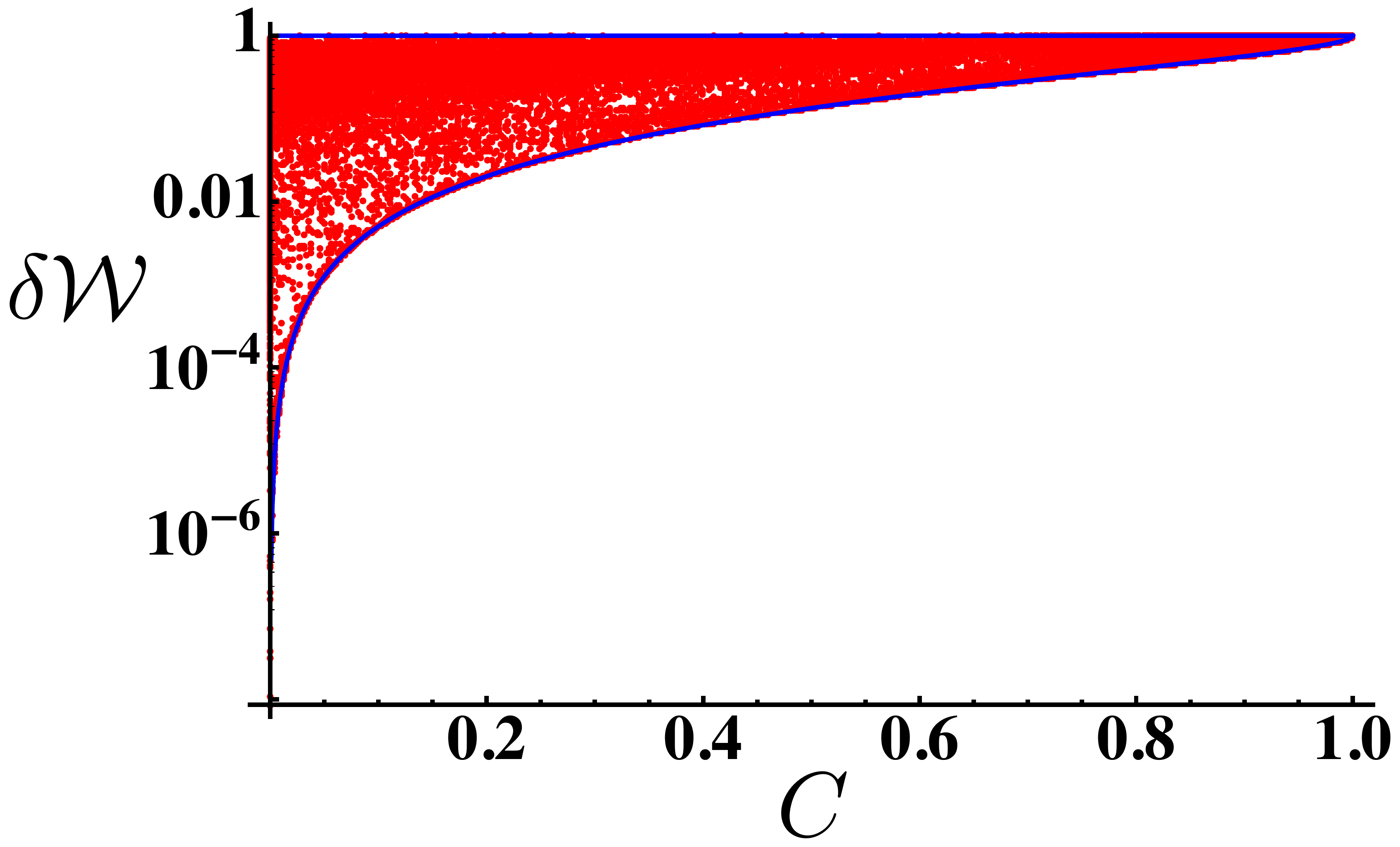}
\caption{(Color online) Distribution of two-qubit states in the
daemonic gain-vs-concurrence plane. We have generated $10^4$
general random states of system and ancilla, evaluating the
concurrence ${\cal C}$ and daemonic gain $\delta{\cal W}$ for each
of them (red dots). The blue curves enclosing the distribution
correspond to boundary families discussed in the body of the
paper. Notice that states at ${\cal C}=0$ may correspond to
arbitrarily large values of the daemonic gain $\delta{\cal W}$.}
        \label{distEnt}
        \end{center}
\end{figure}
We can now address Theorem~\ref{th2} and its consequences for
two-qubit states. Similarly to what was done above, we have
studied the distribution of random two-qubit states in the
daemonic ergotropy-versus-entanglement plane, choosing quantum
concurrence ${\cal C}$ as a measure for the latter~\cite{horo}.
The results are illustrated in Fig.~\ref{distEnt}. As before, a
lower bound to the amount of daemonic ergotropy at set value of
concurrence can be identified. We have that, for any state
$\hat\rho_{SA}$ with concurrence ${\cal C}$
\begin{equation}
\delta{\cal W}\ge\delta{\cal W}_\text{min}({\cal C})=1-\sqrt{1-{\cal C}^2},
\end{equation}
a lower bound that is achieved by Bell-diagonal states that are
fully characterized in the Appendix. The upper bound, on the
other hand, is achieved by maximally ergotropic (in our daemonic
sense) states
$\hat\rho_{SA}=[\ket{00}\bra{00}_{SA}+\ket{11}\bra{11}_{SA}+{\cal
C}(\ket{00}\bra{11}_{SA}+h.c.)]/2$.

\noindent {\it Conclusions.--} We have illustrated an
ancilla-assisted protocol for work extraction that takes advantage
of the sharing of quantum correlations between a system and an
ancilla that is subjected to suitably chosen projective
measurements. Our approach allowed us the introduce of a new form
of information-enhanced ergotropy, which we have dubbed daemonic,
that acts aptly as a witness for quantum correlations in general,
and serves as a necessary and sufficient criterion for
separability of bipartite pure states. We have characterised fully
the distribution of quantum correlated two-qubit states with
respect to the figure of merit set by the daemonic ergotropy,
finding that quantum correlations embody a proper resource for the
work-extraction performances of the states that minimize
$\delta{\cal W}$. Our work opens up interesting avenues for the
thermodynamic interpretation of quantum correlations, clarifies
their resource-role in ancilla-assisted information thermodynamics
and opens up possibilites to understand  the role of correlations
in the charging power of quantum batteries \cite{Binder:15(b)}.

\noindent
{\it Acknowledgements.--} J.~G. would like to sincerely thank F.~Binder, K.~Modi and S.~Vinjanampathy for discussions related to this work. G. Francica thanks the Centre for Theoretical Atomic, Molecular, and Optical Physics, School of Mathematics and Physics, Queen's University Belfast, for hospitality during the completion of this work.  We acknowledge support from the EU FP7 Collaborative Projects
QuProCS and TherMiQ, the John Templeton Foundation (grant number 43467), the Julian Schwinger Foundation (grant number JSF-14-7-0000), and the UK EPSRC (grant number EP/M003019/1). We acknowledge partial support from COST Action MP1209.

\section{Appendix}

Here we define the notion of discord used in the paper, present
details of the core results discussed in the main body of the
paper and the formal proofs of both Theorem~\ref{th1} and
\ref{th2}.

\noindent
{\it Discord.--} We recall the definition of quantum discord $\overleftarrow{\mathcal{D}}_{SA}$ associated with orthogonal measurements $\hat\Phi^A_a$ performed over the ancilla~\cite{discord}
 \begin{equation}
 \label{defdisc}
   \DAtoS= \mathcal{I}_{SA} - \text{max}_{\{\hat\Phi^A_a\}} \overleftarrow{\mathcal{J}}_{SA},
 \end{equation}
where $\mathcal{I}_{SA}$ is the mutual information $\mathcal{I}_{SA}=S(\hat\rho_A)+S(\hat\rho_S)-S(\hat\rho_{SA})$, $\overleftarrow{\mathcal{J}}_{SA} = S(\hat \rho_S)-\sum_a p_a S(\hat \rho_{S|a})$ is the so-called {\it one-way classical information} and $S(\hat\rho)=-\text{Tr}[\hat\rho\log_2\hat\rho]$ is the von Neumann entropy of the general state $\hat\rho$. The maximization inherent in Eq.~\eqref{defdisc} is over all the possible orthogonal measurements on the state of $A$. Similarly we define the discord $\DStoA$ associated with measurements performed over the state of the system $S$ as Eq.~\eqref{defdisc} with the role of $S$ and $A$ being swapped. \\

\noindent
{\it Theorem~\ref{th1}.--} In order to provide a full-fledged assessment of Theorem~\ref{th1}, we should first discuss the following Lemma.
\begin{lemma}
\label{le1} For any set of orthogonal projective measurements
$\{\hat\Pi^A_a\}$ performed over an ancilla $A$ prepared with a
system $S$ in a state $\hat\rho_{SA}$, we have
$\mathcal{W}_{\{\Pi^A_a\}}\geq \mathcal{W}$
\end{lemma}
\begin{proof}
In order to show this statement, we observe that
\begin{equation}
\label{res1}
\begin{aligned}
  \hat\rho_S &= \sum_k r_k \ket{r_k} \bra{r_k} = \sum_a {\rm Tr}_A[\hat\Pi^A_a\hat\rho_{SA}] \\
 &= \sum_a p_a \hat\rho_{S|a} = \sum_a p_a \sum_k r^a_k  \ket{r^a_k} \bra{r^a_k}.
   \end{aligned}
\end{equation}
Eq.~\eqref{res1} implies that $r_k = \sum_a p_a \sum_j  r^a_j\vert{\braket{r_k}{r^a_j}}\vert^2$. As $\mathcal{W}_{\{\hat\Pi^A_a\}} -\mathcal{W} = \sum_k \epsilon_k \left( r_k - \sum_a p_a r^a_k \right)$, we have that
\begin{equation}
   \mathcal{W}_{\{\hat\Pi^A_a\}} -\mathcal{W} = \sum_a p_a \sum_{k,j} r^a_j \epsilon_k \left( \vert{\braket{r_k}{r^a_j}}\vert^2 - \delta_{kj} \right)\geq 0
\end{equation}
due to the fact that $\sum_{k,j} r^a_j \epsilon_k \left( \vert{\braket{r_k}{r^a_j}}\vert^2 - \delta_{k\,j} \right)\geq 0$, as this is the ergotropy of $\hat\rho_{S|a}$ relative to the Hamiltonian $\sum_k \epsilon_k \ket{r_k}\bra{r_k}$.
\end{proof}

We are now in a position to provide the full proof of Theorem~\ref{th1}, which we restate here for easiness of consultation:
\noindent
{\bf Theorem 1.} {\it For any system $S$ and ancilla $A$ prepared in a state $\hat\rho_{SA}$, we have
\begin{equation}
  \delta \mathcal{W} = 0 \Rightarrow  \overrightarrow{\mathcal{D}}_{SA}=0
\end{equation}
with $\delta{\cal W}$ and $\DStoA$ as defined in Eq.~\eqref{daemonergo} and \eqref{DStoA}, respectively.}

\begin{proof}
In light of Lemma~\ref{le1}, we have that
$\mathcal{W}_{\{\hat\Pi^A_a\}} - \mathcal{W} = 0 \Leftrightarrow \delta \mathcal{W}=0$ for any set $\{ \hat\Pi^A_a \}$. Then, in order to proof the statement of the Theorem, it is enough to show that, regardless of the choice of projective set $\{ \hat\Pi^A_a \}, \quad  \mathcal{W}_{\{\hat\Pi^A_a\}} - \mathcal{W} = 0 \Rightarrow \overrightarrow{\mathcal{D}}_{SA}=0$.
 Let assume that $\overrightarrow{\mathcal{D}}_{SA}\neq0$. Then, there is at least a set $\{\hat\Pi^A_a\} $ such that $\mathcal{W}_{\{\hat\Pi^A_a\}} - \mathcal{W} \neq 0$. Two cases are possible:
\begin{enumerate}
\item[{\it (i)}] There is a measurement outcome $\bar{a}$ such that $\hat\rho_{S|\bar{a}} \neq \sum_k r^{\bar{a}}_k \ket{r_k}\bra{r_k}$ with $r^{\bar{a}}_k\geq r^{\bar{a}}_{k+1}$. Then
 \begin{equation}
    \mathcal{W}_{\{\Pi^A_a\}} -\mathcal{W} \geq  p_{\bar{a}} \sum_{k,j}r^{\bar{a}}_j \epsilon_k \left( \vert{\braket{r_k}{r^{\bar{a}}_j}}\vert^2 - \delta_{kj} \right) > 0,
 \end{equation}
given that $\sum_{k,j}r^{\bar{a}}_j \epsilon_k \left( \vert{\braket{r_k}{r^{\bar{a}}_j}}\vert^2 - \delta_{kj} \right) $ is the ergotropy of $\hat\rho_{S|\bar{a}}$ relative to the Hamiltonian $\sum_k \epsilon_k \ket{r_k}\bra{r_k}$, and is zero if and only if $\hat\rho_{S|\bar{a}} = \sum_k r^{\bar{a}}_k \ket{r_k}\bra{r_k}$.

\item[{\it (ii)}] For every $a$ $\rho_{S|a} = \sum_k r^{a}_k \ket{r_k}\bra{r_k}$ with $r^{a}_k\geq r^{a}_{k+1}$. In this case $\mathcal{W}_{\{\Pi^A_a\}} -\mathcal{W} =0$. However, as $\hat\rho_{SA}$ is such that $\overrightarrow{\mathcal{D}}_{SA}\neq0$, it is always possible to identify another set $\{\hat\Pi'^A_a\}$ such that $\mathcal{W}_{\{\Pi'^A_a\}} -\mathcal{W} > 0$. In order to show how this is possible, we note that $\hat\rho_{SA}$ can be written as
 \begin{equation*}
  \hat \rho_{SA} = \sum_{a,a'}\sum_{k,k'} C^{aa'}_{kk'} \ket{r_k}\bra{r_{k'}}_S \otimes \ket{a}\bra{a'}_A
 \end{equation*}
with the condition $p_a C^{aa}_{kk'} = r^a_k \delta_{kk'}$. As $\overrightarrow{\mathcal{D}}_{SA}\neq0$, there are two measurement outcomes $\bar{a}$ and $\bar{a}'$ such that $C^{\bar{a}\bar{a}'}_{kk'} \neq C^{\bar{a}\bar{a}'}_{k} \delta_{kk'}$. Should this be not true, we would have 
$\overrightarrow{\mathcal{D}}_{SA}=0$, and thus a contradiction. Therefore, as $\DStoA\neq0$, the matrix $_A\bra{\bar{a}}\hat\rho_{SA}\ket{\bar{a}'}_A$ cannot be diagonal in the basis $\{\ket{r_k}_S\}$ (here $\ket{\bar{a}_A}$ is the eigenstate of $\hat\Pi^A_a$ with eigenvalues $\bar{a}$). If $\bar{a}=\bar{a}'$, case ${\it (ii)}$ cannot occur.

However, if $\bar{a}\neq \bar{a}'$, we can define the new set of projectors $\{\hat\Pi'^A_a\}$ with elements $\hat\Pi'^A_{\bar{a}} = ( \ket{\bar{a}} +  \ket{\bar{a}'}) ( \bra{\bar{a}} + \bra{\bar{a}'})/2$, $\hat\Pi'^A_{\bar{a}'} =( \ket{\bar{a}} -  \ket{\bar{a}'})( \bra{\bar{a}} - \bra{\bar{a}'})/2$ and  $\hat\Pi'^A_{a}=\hat\Pi^A_a$ for $a\neq \bar{a},\bar{a}'$. Then, the density matrix 
$\hat\rho'_{S|\bar{a}} = {\rm Tr}_A\{\hat\Pi'^A_{\bar{a}} \hat\rho_{SA}\}/p'_{\bar{a}}$ reads
\begin{equation}
\begin{aligned}
 \hat \rho'_{S|\bar{a}} &= \frac{1}{2p'_{\bar{a}}}\Big[\sum_k(p_{\bar{a}}r^{\bar{a}}_k+p_{\bar{a}'}r^{\bar{a}'}_k)\ket{r_k}\bra{r_k}_S \\
  &+ \left(_{A}\!\bra{\bar{a}}\hat\rho_{SA}\ket{\bar{a}'}_A+_{A}\!\bra{\bar{a}'}\hat\rho_{SA}\ket{\bar{a}}_A\right)\Big],
\end{aligned}
\end{equation}
which shows that $\hat\rho'_{S|\bar{a}}$ is not diagonal in the basis $\{\ket{r_k}_S\}$. Therefore $\rho'_{S|\bar{a}} \neq \sum_k r'^{\bar{a}}_k \ket{r_k}\bra{r_k}_S$ with $r'^{\bar{a}}_k\geq r'^{\bar{a}}_{k+1}$. So, proceeding in a similar way as for case ${\it (i)}$, we conclude that
\begin{equation}
  \mathcal{W}_{\{\Pi'^A_a\}} -\mathcal{W} > 0.
\end{equation}
If $_A\!\bra{\bar{a}}\hat\rho_{SA}\ket{\bar{a}'}_A+_A\!\bra{\bar{a}'}\hat\rho_{SA}\ket{\bar{a}}_A = 0$, it is enough to consider $\hat\rho'_{S|\bar{a}'}$ instead of $\hat\rho'_{S|\bar{a}}$.
\end{enumerate}
\end{proof}

Having proven Theorem~\ref{th1}, we can provide a justification of two important Corollaries

\begin{corollary}
Under the premises of Theorem~\ref{th1}, $\delta \mathcal{W} = 0 \nRightarrow  \protect\overleftarrow{\mathcal{D}}_{SA}=0$.
\end{corollary}
\begin{proof}
It is enough to consider the state
\begin{equation}
\hat \rho_{SA}=\sum_{k,a} q_{ak} \ket{r_k}\bra{r_k}_S \otimes \ket{\phi_a}\bra{\phi_a}_A,
\end{equation}
where $\{\ket{\phi_a}_A\}$ is a non orthogonal set of states.
Under such conditions, we have
$\overleftarrow{\mathcal{D}}_{SA}\neq0$. If we choose $q_{ak}$
such that $q_{ak} \geq q_{ak+1}$, we have
$\mathcal{W}_{\{\hat\Pi^A_a\}} - \mathcal{W}=0$ for any set
$\{\hat\Pi^A_a\}$, as $\hat\rho_{S|a}=\sum_k r^a_k
\ket{r_k}\bra{r_k}_S$ with $r^a_k=\sum_{a'}q_{a'\,k}
\abs{\braket{\phi_{a'}}{a_A}}^2/p_a\geq r^a_{k+1}$).
\end{proof}

\begin{corollary}
\label{coroll2}
Under the premises of Theorem~\ref{th1}, we have that $\protect\overleftarrow{\mathcal{D}}_{SA}=0 \text{ or } \protect\overrightarrow{\mathcal{D}}_{SA}=0 \nRightarrow   \delta \mathcal{W} = 0$.
\end{corollary}
\begin{proof}
We consider the state $\hat\rho_{SA} = \sum_k r_k \hat\Pi^S_k \otimes \hat\Pi^A_k$, where $\hat\Pi^A_k$ and $\hat\Pi^S_k$ are orthogonal projectors of rank one. Although such state has zero discord,  the quantity $\mathcal{W}_{\{\Pi^A_k\}} - \mathcal{W}$ is positive since
\begin{equation}
  \mathcal{W}_{\{\Pi^A_k\}} - \mathcal{W} = \sum_k r_k \epsilon_k - \epsilon_1 >0.
\end{equation}
Therefore, $\delta \mathcal{W} >0$.
\end{proof}

\noindent
{\it Theorem 2.--} We can now provide a proof of Theorem~\ref{th2}, which we state again for easiness of consultation.

\noindent
{\bf Theorem 2.} {\it For any system $S$ and ancilla $A$ prepared in a pure state $\hat \rho_{SA}=\ket{\psi}\bra{\psi}_{SA}$ we have
\begin{equation}
\delta{\cal W}=0\Leftrightarrow\ket{\psi}_{SA}~\text{is separable},
\end{equation}
and $\delta \mathcal{W} = \sum_k r_k \epsilon_k - \epsilon_1$, where $r_k$ are the Schmidt coefficients of $\ket{\psi}_{SA}$ and $\epsilon_k$ are the eigenvalues of $\hat H_S$, ordered such that $r_k\geq r_{k+1}$ and $\epsilon_k\leq \epsilon_{k+1}$.}
\begin{proof}
We make use of the instrumental result embodied by Corollary~\ref{coroll2} and consider the pure state $\hat\rho_{SA}=\ket{\psi_{SA}}\bra{\psi_{SA}}$ whose Schmidt decomposition reads $\ket{\psi_{SA}} = \sum_k \sqrt{r_k} \ket{r_k}_S\otimes\ket{\phi_k}_A$ with $r_k\geq r_{k+1}$. Corollary~\ref{coroll2} has shown that $\delta \mathcal{W} = \sum_k r_k \epsilon_k - \epsilon_1$. Therefore, $\delta{\cal W}=0$ it must be $\epsilon_1=\sum_k r_k \epsilon_k$, which implies $r_k=\delta_{1k}$. This implies that the state has a single Schmidt coefficient, and is thus separable. The proof of the reverse statement is trivial.

\end{proof}

\noindent
{\it Analysis of the two-qubit case.--} We provide additional details on the analysis performed on the two-qubit case illustrated in the main body of the paper.

In what follows, with no loss of generality, we choose the system Hamiltonian $\hat H_S=-\hat\sigma_z$. As stated in the main body of the paper, we choose concurrence as the entanglement measure to be used in our analysis. 
For a bipartite qubit state, concurrence is defined as~\cite{horo}
\begin{equation}
  \mathcal{C} = \text{max}[0,\lambda_1-\sum_{j>1}\lambda_j],
\end{equation}
where $\lambda_k$ are the square roots of the eigenvalues of $\hat
\rho \hat{\tilde\rho}$ with $\hat{\tilde\rho}=(\hat
\sigma_y\otimes\hat\sigma_y)\hat\rho^*(\hat\sigma_y\otimes\hat\sigma_y)$,
ordered so that $\lambda_k\geq \lambda_{k+1}$. In the main body of
the paper we have proven that the ergotropic gain of any state
$\hat\rho_{SA}$ with concurrence $\mathcal{C}$ is larger than, or
equal to
\begin{equation}
\label{eq dmin}
\delta \mathcal{W}_\text{min}(\mathcal{C})=1-\sqrt{1-{\cal C}^2}.
\end{equation}
%
%
%

The states locally equivalent to
\begin{equation} \label{state1}
  \hat\rho_{SA}=\left(
              \begin{array}{cccc}
                0 & 0 & 0 & 0 \\
                0 & x & \mathcal{C}/2 & 0 \\
                0 & \mathcal{C}/2 & 1-x & 0 \\
                0 & 0 & 0 & 0 \\
              \end{array}
            \right)
\end{equation}
with $x=(1\pm\sqrt{1-\mathcal{C}^2})/2$, which have concurrence $\mathcal{C}$, are such that $\delta \mathcal{W}=\delta \mathcal{W}_\text{min}(\mathcal{C})$. These states belong to the class parametrized as $p \ket{\phi_+^\eta}\bra{\phi_+^\eta}+\frac{1-p}{2}(\ket{01}\bra{01}+\ket{10}\bra{10})$ where $\ket{\phi_+^\eta}=\sqrt{\eta}\ket{01}+\sqrt{1-\eta}\ket{10}$. On the other hand, the states locally equivalent to
\begin{equation} \label{state2}
\hat  \rho_{SA}=\left(
              \begin{array}{cccc}
                1/2 & 0 & 0 & \mathcal{C}/2 \\
                0 & 0 & 0 & 0 \\
                0 & 0 & 0 & 0 \\
                \mathcal{C}/2 & 0 & 0 & 1/2 \\
              \end{array}
            \right),
\end{equation}
which also have concurrence $\mathcal{C}$, are such that $\delta
\mathcal{W}=1$, and thus embody the upper bound to the daemonic
ergotropy at set value of concurrence.

In order to show this, we parameterize the projectors
$\hat\Pi^A_1$ and $\hat\Pi^A_2$ that are needed to calculate the
daemonic ergotropy in terms of the angles $\theta \in [0,\pi]$ and
$\phi\in[0,2\pi)$ such that
\begin{equation*}
  \Pi^A_1 = \left(
              \begin{array}{cc}
                \cos^2\left({\theta}/{2}\right) & e^{-i\phi}{\sin(\theta}/{2}) \\
                e^{i\phi}{\sin(\theta}/{2}) & \sin^2\left({\theta}/{2}\right)  \\
              \end{array}
            \right)
\end{equation*}
and $\Pi^A_2=\openone-\Pi^A_1$. An extensive numerical analysis of
the distribution itself has shown that the states lying on the
lower boundary belong to the class of so-called {\it $x$-states}
of the form
\begin{equation}
  \rho_{SA} = \left(
                \begin{array}{cccc}
                  a & 0 & 0 & z \\
                  0 & b & w & 0 \\
                  0 & w & c & 0 \\
                  z & 0 & 0 & 1-a-b-c \\
                \end{array}
              \right),
\end{equation}
where $a,b,c,w,z$ are positive numbers such that $bc\geq w^2$,
$ad\geq z^2$. This class plays a key role in the characterisation
of the states that maximize quantum correlations at set values of
the purity of a given bipartite qubit state~\cite{girolami,MEMS}.
The ergotropy $\mathcal{W}$ for such class of states is
\begin{equation}
  \mathcal{W} = 
  \begin{cases}
                  0 &\text{ for } a+b \geq \frac{1}{2}, \\
                  2-4(a+b) &\text{ otherwise}.
\end{cases}
\end{equation}
On the other hand, we have $\mathcal{W}_{\{\Pi^A_a\}} = 1-2(a+b)+({X_++X_-})/{2}$ with
\begin{equation}
\begin{aligned}
  X_\pm &=\left\{\left[2(a+b)-1\pm(1-2b-2c)\cos\theta\right]^2\right. \\
   &\left.+4\left[we^{-i\phi}+ze^{i\phi}\right|^2\sin^2\theta\right\}^{\frac{1}{2}}.
\end{aligned}
\end{equation}
The associated concurrence is $\mathcal{C} =
2\text{max}\{0,z-\sqrt{bc},w-\sqrt{ad}\}$. We make the ansatz that
a state as in Eq.~\eqref{state1} with $x$ real and positive,
minimizes $\delta \mathcal{W}$ at a fixed value of $\mathcal{C}$.
Then, from the positivity of the density matrix, $x$ must satisfy
the condition $\mathcal{C}\leq 2 \sqrt{x(1-x)}$ with $x\in[0,1]$.

For such state, we have
$\delta\mathcal{W}=2-2x-\text{max}\{0,2-4x\}$. If we consider
$x\geq1/2$, then $\delta\mathcal{W}=2-2x$, which is minimum when
$x$ is maximum, i.e. for $x=(1+\sqrt{1-\mathcal{C}^2})/2$. For
$x\leq1/2$ we have $\delta\mathcal{W}=2x$, which is minimum when
$x$ is minimum, i.e. for $x=(1-\sqrt{1-\mathcal{C}^2})/2$. In both
cases, $\delta \mathcal{W}$ takes the expression in Eq.~\eqref{eq
dmin}.

In order to show that the class in Eq.~\eqref{state2} is such that
$\delta \mathcal{W}=1$, it is enough to observe that, for such
state, $\mathcal{W}=0$. In fact, we trivially have $\rho_S =
\mathbf{1}/2$ and, by choosing for instance $\hat\Pi^A_1 =
\ket{0}\bra{0}_A$, we get pure post-measurement states, and thus
$\mathcal{W}_{\{\hat\Pi^A_a\}} = 1$. Therefore $\delta
\mathcal{W}=1$ regardless of the value taken by $\mathcal{C}$.

As mentioned above, the validity of the ansatz used here is
justified by an extensive numerical investigation based on $10^6$
random bipartite states generated uniformly according to the Haar
measure.

\end{document}